\documentclass[pdflatex,sn-vancouver-num]{sn-jnl}

\usepackage{graphicx}%
\usepackage{multirow}%
\usepackage{amsmath,amssymb,amsfonts}%
\usepackage{amsthm}%
\usepackage{mathrsfs}%
\usepackage[title]{appendix}%
\usepackage{xcolor}%
\usepackage{textcomp}%
\usepackage{manyfoot}%
\usepackage{booktabs}%
\usepackage{algorithm}%
\usepackage{algorithmicx}%
\usepackage{algpseudocode}%
\usepackage{listings}%
\usepackage{enumitem}%
\usepackage{tikz}
\usetikzlibrary{arrows.meta,positioning,fit,shapes.geometric}
\usepackage{anyfontsize}

\newtheorem{theorem}{Theorem}

\newcommand{\bM}{\boldsymbol{M}} 
\newcommand{\bN}{\boldsymbol{N}}
\newcommand{\bW}{\boldsymbol{W}} 
\newcommand{\bZ}{\boldsymbol{Z}} 
\newcommand{\R}{\mathbb{R}} 
\renewcommand{\P}{{\mathbb{P}}} 

\begin{document}

\title[Article Title]{Considering causality in the construction of molecular signatures of lifestyle exposures} 

\author[1]{\fnm{Diana} \sur{Wu}}\email{wud@iarc.who.int}

\author*[1]{\fnm{Vivian} \sur{Viallon}}\email{viallonv@iarc.who.int}

\affil[1]{\orgdiv{Nutrition and Metabolism Branch}, \orgname{International Agency for Research on Cancer (IARC/WHO)}, \orgaddress{\city{Lyon}, \country{France}}}

\maketitle
{\small
\setlength{\parindent}{0pt}
\setlength{\parskip}{4pt}
Declarations:
Where authors are identified as personnel of the International Agency for Research on Cancer/World Health Organization, the authors alone are responsible for the views expressed in this article and they do not necessarily represent the decisions, policy, or views of the International Agency for Research on Cancer/World Health Organization.
\\
Ethics approval and consent to participate:
Not applicable
\\
Consent for publication:
Not applicable

Availability of data and materials:
The code to generate the data and for analyses in this study is available at: https://github.com/IARCBiostat/SimulationSignatures

Competing interests: The authors have no competing interests to declare that are relevant to the content of this article.

Funding:
This work was financially supported in part by funding for grant IIG\_FULL\_2022\_013 which was obtained from Wereld Kanker Onderzoek Fonds (WKOF) as part of the World Cancer Research Fund International grant programme. This work was also supported in part by the EU Marie Curie Doctoral Training Network "ColoMARK", grant agreement no. 101072448.

Authors' contributions:
 Conceptualization: V.V.; Methodology: V.V., D.W.; Formal analysis and investigation: V.V.,D.W.; Writing - original draft preparation: V.V.,D.W.; Writing - review and editing: V.V.,D.W.; Funding acquisition: V.V.; Resources: V.V.; Supervision: V.V.

Acknowledgements:
Not applicable
}

\newpage

\section*{Abstract}
Molecular signatures derived from omics data are increasingly used in epidemiological studies to characterize lifestyle exposures, either as proxies of exposure or to provide insight into disease mechanisms. These signatures are typically constructed by regressing the exposure on high-dimensional omics features. In the literature, an initial univariate screening step has sometimes been applied prior to multivariate modelling, but the causal implications of this choice have not yet been considered. 
Focusing on settings where the exposure causally influences molecular features (and not the reverse), we use directed acyclic graphs (DAGs) and $d$-separation arguments to show that collider bias may arise when the screening step is ignored, leading to the inclusion of non-causal features in the signature. We further demonstrate that the screening step can mitigate this bias.
Our simulation studies illustrate that screening reduces the inclusion of non-causal features, albeit at the cost of lower sensitivity and reduced correlation between the exposure and the resulting signature. Overall, we recommend applying univariate screening prior to signature construction, particularly when the inclusion of non-causal features is undesirable, such as in mechanistic studies.

\newpage

\section{Background}\label{sec1}
As omics data, including metabolomics and proteomics, become more common in epidemiological studies, biomarker discovery through omics has emerged as an expanding field of interest \cite{brennan_metabolomics_2021}. Omics-based biomarkers often involve data-driven discovery through the construction of signatures, usually as a weighted sum of individual biomarkers, for example to represent a lifestyle exposure \cite{smith_healthy_2022,watanabe_multiomic_2023,shah_dietary_2023,li_development_2024,perry_proteomic_2024,loftfield_novel_2021,assi_metabolic_2018,assi_are_2018,wan_plasma_2025,menni_metabolomic_2017}.

These “molecular signatures” usually serve either as a composite biomarker to more accurately characterize lifestyle exposures or are used to investigate mechanisms of their health consequences \cite{scalbert_food_2019,brennan_metabolomics_2021}.    
As biomarkers of lifestyle exposures, signatures are intended to provide objective and reproducible surrogates to self-reported lifestyle assessments, which might help to mitigate measurement-error challenges associated with the latter \cite{cuparencu_towards_2024,smith_healthy_2022,watanabe_multiomic_2023,shah_dietary_2023,li_development_2024,perry_proteomic_2024}.  Omic signatures are also used to uncover biological mechanisms affected by lifestyle exposures potentially related to the etiology of chronic diseases \cite{scalbert_food_2019,loftfield_novel_2021,assi_metabolic_2018,assi_are_2018,wan_plasma_2025,menni_metabolomic_2017}. In the latter case (and in some respect in the former case as well; see the Discussion for more details), signatures should include features that are causally influenced by the exposure only.

Multivariate statistical approaches, such as penalized linear regression models (e.g., lasso or elastic-net), are commonly used for the derivation of molecular signatures  \cite{waldron_optimized_2011}. These approaches typically use the entire set of available features as the input (hereafter referred to as the no-screening strategy)  \cite{assi_metabolic_2018,smith_healthy_2022,watanabe_multiomic_2023,shah_dietary_2023,li_development_2024,perry_proteomic_2024,wan_plasma_2025}
, although some authors have restricted input to features showing statistically significant associations with the exposure in an initial univariate analysis (hereafter referred to as the screening strategy) \cite{menni_metabolomic_2017,cirulli_profound_2019,zhu_proteomic_2025}. 
Previous studies have used either one of these two strategies to derive novel biomarkers, draw mechanistic conclusions, or both \cite{cirulli_profound_2019,zhu_proteomic_2025}. Although screening steps can be applied for computational efficiency \cite{fan_sure_2008}, to our knowledge there is no clear indication of the relative merits of the two strategies  for the construction of signatures and their implications in the interpretation of signature-based results in the literature.  

In this article, we address this gap using arguments based on directed acyclic graphs (DAGs) and results from simulation studies. We focus on the setting in which the lifestyle exposure may causally influence the molecular features and not vice versa, which is the usual assumption in most previous work on molecular signatures of lifestyle exposures \cite{rattray_beyond_2018, walker_metabolome_2019}. In Section \ref{sec:Methods}, we formally establish that the no-screening strategy can induce collider bias, resulting in the inclusion of molecular features spuriously associated with, rather than causally influenced by, the exposure in the signature, which is undesirable particularly when the goal is to investigate biological processes influenced by the exposure. We further show that the screening strategy does typically not suffer from this limitation. Empirical results on simulated data are presented in Section \ref{sec:Results} to illustrate the performance of the screening and no-screening strategies on finite samples. We conclude with a discussion in Section \ref{sec:Disc}.

\section{Methods} \label{sec:Methods}

\subsection{Notation and assumptions}
Denote by $E$ the lifestyle exposure of interest and by $\bM = (M_1,\dots,M_p)$ a vector of $p$ molecular features. In most practical settings, $(M_1,\dots,M_p)$ exhibit complex correlation patterns, reflecting both direct causal relationships among the features and shared underlying causes. Accordingly, we denote by $W = (W_1, \dots, W_d)$, for some $d\geq 1$, the vector of observed shared causes of $E$ and $(M_1,\dots,M_p)$, which we will suppose to be fully observed. We further denote by $U = (\nu_1,\dots,\nu_q)$, for some $q\geq 1$, the set of other causes of $M_1,\dots,M_p$, which are not causes of exposure $E$ and are usually unobserved or not accounted for in the analysis. 

Let \(G\) be the DAG over nodes \(\{E,M_1,\dots,M_p,W_1, \dots, W_d, \nu_1,\dots,\nu_q\}\). We denote by $\mathrm{De}(E)$ the descendants of $E$, i.e. 
\[
\mathrm{De}(E):=\{\,M_j: {\rm a\ directed\ path\ of\ the\ form\ }E\to \dots \to M_j \text{ exists in } G\,\}
\] and by $\mathrm{Ch}(E) \subseteq \mathrm{De}(E)$ the set of children of \(E\), 
\[
\mathrm{Ch}(E):=\{\,M_j: \text{an edge }E\to M_j \text{ exists in } G\,\},
\]

In practice, some relevant molecular features may be unobserved. We denote by $\bN$ the vector of size $p_0$ comprising the observed molecular features. Without loss of generality, we assume that the $p_0$ observed features are the first $p_0$ features in $\bM$. For any vector $\bZ\in\R^K$,  we let $\bZ_S$ and $\bZ_{-S}$ be the vector of features indexed by $S$ and $\{1, \dots, K\} \setminus S$ for  any set $S \subset \{1, \dots, K\}$. Finally, we will at times use the notation $\bM$ and $\bN$  to denote the set of features in vectors $\bM$ and $\bN$, respectively.  

Throughout, we assume that the joint distribution \(\P\) of $\{E,M_1,\dots,M_p,$ $W_1, \dots, W_d,$ $\nu_1,\dots,\nu_q\}$ is Markov and faithful with respect to \(G\), which ensures that all and only the conditional independences in the distribution correspond to d-separation relations in the DAG \citep{pearl_causality_2009,  peters2017elements}. In addition, we assume that $E$ is a cause of at least one of these observed features.

Figure \ref{fig:Dagtoy} presents the DAG of a toy example for illustration. Here, $p=19$, $d=1$, $q=2$, \( \mathrm{De}(E) = \{M_1, M_3, M_4, M_6, M_8, M_{11}, M_{12}, M_{13}, M_{17}, M_{18}\}\), \( \mathrm{Ch}(E) = \{M_1,  M_6, M_{11}, M_{17}\}\), and $\nu_1, \nu_2$ but also $M_{17}$, are not observed ($p_0=18$).

\begin{figure}
\begin{tikzpicture}[
    >=Latex,
    node distance=1.6cm and 2.0cm,
    expo/.style   ={draw, circle, minimum width=9mm, minimum height=7mm, fill=blue!20, text=black, thick},
    child/.style  ={draw, circle, minimum size=8mm, fill=blue!05, text=black, thick},
    childlat/.style  ={draw, diamond, minimum size=8mm, fill=blue!05, text=black, thick, dashed},
    desc/.style   ={draw, circle, minimum size=8mm, fill=green!10, text=black, thick},
    confoun/.style = {draw, circle, minimum size=8mm, text=black, thick},
    inR/.style    ={draw, circle, aspect=1.3, minimum size=9mm, fill=red!25, text=black, thick}, 
    nondesc/.style={draw, circle, minimum size=8mm, fill=red!25, text=black, thick},
    latent/.style ={draw, diamond, minimum size=8mm, fill=gray!10, text=black, dashed},
    lab/.style    ={draw=none, align=left}
]

\node[expo] (E) {E};

\node[child, right=2.5cm of E, yshift= 2.5cm] (C1) {$M_1$};
\node[child, right=2.5cm of E]               (C2) {$M_6$};
\node[child, right=2.5cm of E, yshift=-2.5cm] (C3) {$M_{11}$}; 

\node[desc, right=1.6cm of C1] (D1) {$M_3$};
\node[desc, right=1.6cm of D1] (D2) {$M_4$};

\node[desc, right=2.2cm of C2, yshift=0.75cm] (D3) {$M_8$};

\node[desc, right=2.2cm of C3] (D4) {$M_{12}$};
\node[desc, right=1.85cm of D4] (D5) {$M_{13}$};

\node[inR,  above=0.5cm of C1, xshift=0.5cm] (P1) {$M_2$};
\node[inR, above=0.5cm of C2, xshift=0.5cm] (P2) {$M_7$};

\node[latent, right=1.3cm of C2, yshift=-1cm] (U1) {$\nu_1$};
\node[inR,     right=1.6cm of U1]              (S1) {$M_9$};

\node[latent,right=1.3cm of C3, yshift=-1.7cm] (U2) {$\nu_2$};
\node[inR,     right=1.6cm of U2]              (S2) {$M_{14}$};
\node[inR, right=1.6cm of S2] (S3) {$M_{15}$};
\node[nondesc, right=1.6cm of S1] (N1) {$M_{10}$};
\node[nondesc, right=1.6cm of D2]               (N3) {$M_5$};

\node[confoun, below=1.5cm of E] (W) {$W_1$};
\node[nondesc, below=4cm of E] (M16) {$M_{16}$};

\node[childlat, right=2.5cm of E, yshift=-4.96cm] (unobs) {$M_{17}$};

\node[desc,right=2.2cm of unobs, yshift=-0.5cm, fill=green!10] (M18) {$M_{18}$};
\node[inR,right=2.2cm of unobs, yshift = -2cm] (M19) {$M_{19}$};

\draw[->] (E) -- (C1);
\draw[->] (E) -- (C2);
\draw[->] (E) -- (C3);
\draw[->, dotted] (E) -- (unobs);

\draw[->] (C1) -- (D1);
\draw[->] (D1) -- (D2);

\draw[->] (C2) -- (D3);

\draw[->] (C3) -- (D4);
\draw[->] (D4) -- (D5);

\draw[->, dotted] (unobs) -- (M18);

\draw[->] (D1) ..controls(4.5, 0.5).. (C2);

\draw[->] (P1) -- (C1);
\draw[->] (P2) -- (C2);
\draw[->, dotted] (M19) -- (unobs);

\draw[->] (U1) -- (C2);
\draw[->] (U1) -- (S1);

\draw[->] (U2) -- (C3);
\draw[->] (U2) -- (S2);
\draw[->] (S3) -- (S2);

\draw[->] (S1) -- (N1);
\draw[->] (N3) -- (D2);
\draw[->] (W) -- (E);
\draw[->] (W) -- (M16);
\draw[->] (W) -- (C3);

\node[draw, rounded corners, inner sep=6pt, fit={(E) (P1) (P2) (N1) (S3) (M19)}] (dagframe) {};
\end{tikzpicture}

\caption{ Directed acyclic graph (DAG) illustrating the causal relationships between a lifestyle exposure $E$ and molecular features $M_1,\ldots,M_{19}$ in our toy example. The exposure $E$ has a direct effect on some features (children, light blue) and an indirect effects on downstream descendants (green). The graph also includes non-descendants of $E$ (red), shared causes of $E$ and some features (white; $W_1$) other shared causes of features (typically unobserved, light grey; $\nu_1$ and $nu_2$). Dashed diamonds represent unobserved variables ($\nu$'s, as well as the child $M_{17}$ in this example).}\label{fig:Dagtoy}
\end{figure}
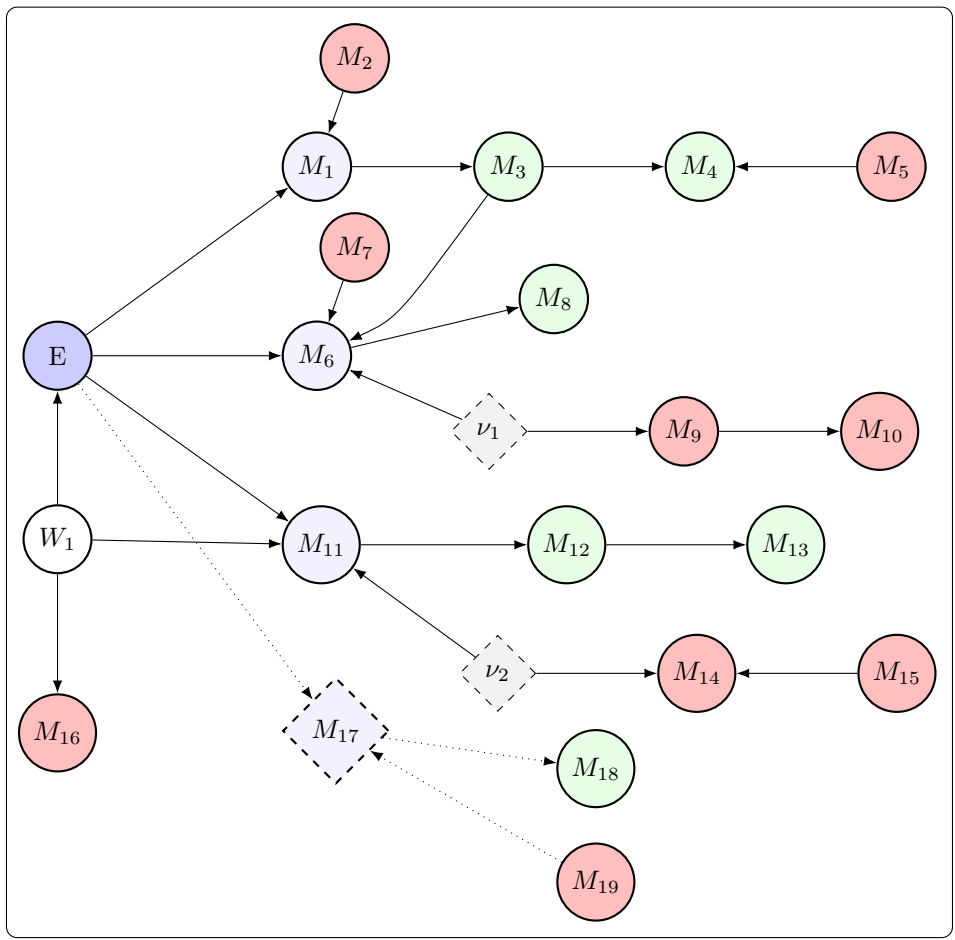

\subsection{Intuition}
As exemplified in Figure \ref{fig:Dagtoy}, the DAG $G$ typically contains the following motifs
\begin{align*}
& E \to M_i, \qquad\qquad\qquad\quad \,\ \  E \to M_i \to \dots \to M_k, \\ 
& E  \to  \dots \to M_k \leftarrow M_l, \qquad 
E  \to  \dots \to M_k \leftarrow U \to M_\ell.   
\end{align*}

For simplicity, first consider each motif as an example of an elementary DAG. The following statements are direct consequences of the notion of d-separation in DAGs under our assumption that the distribution $\P$ is Markov and faithful to the DAG \cite{pearl_causality_2009}. In the first two DAGs, $M_i$ is directly influenced by $E$ (child of $E$), so that $M_i$ is associated with $E$, both marginally and upon conditioning on the other features. In the last three examples, $M_k$ is a descendant of $E$. It is associated with $E$ marginally, but is not associated with $E$ after conditioning on any of its ancestors, such as $M_i$ in the second example. In the third and fourth examples, $M_l$ and $M_\ell$ are not influenced by $E$ and both are marginally non-associated with $E$. However, $M_k$ is a collider on these paths, implying that conditioning on $M_k$ induces a spurious, non-causal, association between $E$ and $M_l$ (third example) or between $E$ and $M_\ell$ (fourth example). If $M_k$ is not observed (such as $M_{17}$ in the DAG of Figure \ref{fig:Dagtoy}), $M_l$ and $M_\ell$ remain  marginally non-associated with $E$ and become associated with $E$ upon conditioning on descendants on $M_k$, or any of its ancestors in $\mathrm{De}(E)$. For example, in the DAG of Figure \ref{fig:Dagtoy},  $M_{19}$ is associated $E$ when conditioning on $\{M_{18}\}$, since $\{M_{18}\}$ is a descendant of the (unobserved) collider $M_{17}$.

In the literature, molecular signatures have often been derived based on multivariate regression models of $E$ on the available molecular features, $\bN$ \cite{smith_healthy_2022, watanabe_multiomic_2023, shah_dietary_2023, li_development_2024, perry_proteomic_2024, menni_metabolomic_2017,zhu_proteomic_2025}. These models estimate the association between $E$ and each feature included in the model, conditional on all other variables included in the models. The four simple examples described in the previous paragraph illustrate that including a child or descendant of $E$ in the model can induce collider bias, thereby creating associations between $E$ and variables that are not its descendants, such as $M_l$ and $M_\ell$ above. Theorem \ref{thm:no-screening} below more formally establishes that (penalized) multivariate regressions based on all available features (no-screening strategy) typically select non-descendants of $E$ in the signature. In the example of Figure \ref{fig:Dagtoy}, features selected in the signature based on the no-screening strategy would asymptotically be \(S_{{\rm no-s}}=\{M_1, M_2, M_3, M_6, M_7, M_9, M_{11}, M_{14}, M_{15}, M_{18}, M_{19}\}\). In particular, the signature would include the non-descendants $M_2, M_7, M_9, M_{14}, M_{15}$ and $M_{19}$. 

In Section \ref{sec:Formal}, we also formally establish that the screening step based on univariate analyses of $E$ and each $M_j$, for $1\leq j\leq p$, adjusted for $W$, asymptotically corrects for this defect. Indeed, the screening step selects exactly the observed descendants of $E$, asymptotically. Obviously, the multivariate model restricted to observed descendants cannot include non-descendants. In the example of Figure \ref{fig:Dagtoy},  it can be shown that features selected in the signature based on the screening strategy is asymptotically $S_{\rm s} = \{M_1, M_3, M_6, M_{11}, M_{18}\}$ (keeping in mind that $M_{17}$ is unobserved).




\subsection{Formal arguments}\label{sec:Formal}

\subsubsection{The no-screening and screening strategies}\label{sec:ScreenNoScreen}

Here, we focus on versions of the no-screening and screening strategies that adjust for confounders $\bW$. The Discussion provides details on the limitations and possible advantages of their unadjusted counterparts.  

As mentioned above, the most standard approach for the derivation of molecular signatures of lifestyle exposures is (penalized) multivariate regression. Under the no-screening strategy, all available molecular features $\bN$ are used, and the model assumes the existence of parameters $\alpha^*\in\R$, $\beta^*\in\R^{p_0}$ and $\theta^*\in\R^d$ such that
\begin{equation}
E = \alpha_0^* + \bN^T\alpha^* + \bW^T\theta^* + \varepsilon,
\label{eq:multivar_reg}
\end{equation}
where $\varepsilon$ is a centered random noise variable. Under this model, the theoretical molecular signature is the weighted average $\bN^T\alpha^*$ of the features in $\bN$. 

Under suitable regularity conditions, and for appropriate choices of the tuning parameter, the LASSO is variable-selection consistent, in the sense it correctly identifies the set of non-zero coefficients in $\alpha^*$ with probability tending to one as the sample size increases \cite{tibshirani_regression_1996,buhlmann_statistics_2011}. Given a training sample of $n$ observations, the LASSO consists in solving the optimization problem 
\begin{equation} \min_{\substack{\alpha_0\in\R, \alpha\in\R^{p_0}\\\theta\in \R^d}} \sum_{i=1}^n \left[E_i - (\alpha_0 + \bN_i^T\alpha + \bW_i^T\theta)\right]^2 + \lambda_n \|\alpha\|_1, \label{eq:lasso_unscreen}
\end{equation}
with $\|\alpha\|_1$ denoting the $L_1$-norm of vector $\alpha$ and $\lambda_n > 0 $ the tuning parameter. Although alternative strategies are possible, the approach presented above consists in not penalizing the parameters corresponding to the confounders $\bW$, ensuring that they are included in the model and properly adjusted for.



Under the screening strategy, the first step considers one molecular feature at a time, using linear models adjusted for confounders $\bW$, but not for the other molecular features. Specifically, for any $j = 1, \dots, p_0$, these models are of the form 
\begin{equation*}
E = \gamma_j^* + \delta_j^*\bN_j + \bW^T\rho_j^* + \epsilon_j,
\end{equation*}
where $\epsilon_j$ is a centered random noise variable and $\gamma_j^*\in\R$, $\delta_j^*\in\R$ and $\rho_j^*\in\R^d$ are the model parameters. For all $j=1, \dots, p_0$, let $H_0(j)$ be the null hypothesis for feature $j$, $H_0(j): \delta^*_j = 0$. Given a sample of size $n$, standard hypothesis tests (Wald or likelihood-ratio), followed by a multiple-testing correction, such as Bonferroni or the Benjamini-Hochberg false-discovery-rate procedure, can be used to derive the set ${\cal R}=\{j: H_0(j) \mathrm{\ is\  rejected}\}$. Then, with $\tilde \bN =\bN_{{\cal R}}$ the vector of features kept by the screening step, the second step of the screening strategy consists of estimating the (penalized) multivariate regression model, as in the no-screening strategy, except $\bN$ is replaced by $\tilde \bN$  (e.g., in \eqref{eq:multivar_reg} and \eqref{eq:lasso_unscreen} above).

\subsubsection{General results}
We can now state the following two general results, which allow the formal assessment of the asymptotic behavior of the no-screening and screening strategies, respectively, assuming they both rely on variable-selection consistent methods. 

\begin{theorem}\label{thm:no-screening}
Assume that some observed non-descendant $M$ of $E$ satisfies one of the two following conditions
\begin{itemize}
    \item[(C.1)] $G$ contains the path \(E \to \dots \to D \leftarrow M\) for some $D\in \mathrm{De}(E)$
    \item[(C.2)] $G$ contains the path \(E \to \dots \to D \leftarrow \nu \to M\) for some $D\in  \mathrm{De}(E)$ and $\nu \in U$
\end{itemize}
Then, under the no-screening strategy, the signature asymptotically includes the non-descendant $M$ whenever $D$ is observed, or has observed descendants, or has observed ancestors in $\mathrm{De}(E)$.
\end{theorem}

\begin{theorem}\label{thm:screening}
Under the screening strategy, the features selected in the signature asymptotically constitute a subset of $\mathrm{De}(E)$, and in particular include all observed children of $E$. 
\end{theorem}

Theorem \ref{thm:no-screening} formally states that the no-screening strategy typically suffers from collider bias. As a result, it asymptotically yields signatures that include non-descendants of $E$, in particular those satisfying conditions $(C.1)$ or $(C.2)$. In fact, additional non-descendants of $E$ may also be included. By arguments similar to those used in the proof of Theorem \ref{thm:no-screening}, it follows that if $M_1$ is selected by the no-screening strategy, then any feature $M_2$ connected to $M_1$ through a path $M_1 \leftarrow  M_2$ or $M_1 \leftarrow U \to M_2$ in $G$ would also be selected asymptotically. This ``recursive" property explains why $M_{15}$ is asymptotically selected by the no-screening strategy in the example of Figure \ref{fig:Dagtoy}. 

Conversely, Theorem \ref{thm:screening} states that the screening strategy asymptotically yields signatures that include only descendants of $E$. While paths of the form $E \to C \to D$ are blocked after adjustment for $C$, a descendant $D$ that is not a direct child may still be selected if, for example, $D$ is a cause of another child of $E$, $E \to C'  \leftarrow D$. For instance, in Figure \ref{fig:Dagtoy},  $M_3$ is a descendant but not a direct child of $E$; it would still be selected asymptotically under the screening strategy because of the path $E \to M_6 \leftarrow M_3$.

With the notation of Theorem \ref{thm:no-screening}, the non‑descendant features $M$ selected under the no‑screening strategy carry predictive information about $E$ after conditioning on $D$. In other words, a signature derived under the no‑screening strategy is expected to achieve better predictive accuracy than one obtained under the screening strategy (as confirmed by our empirical results below). However, the inclusion of such features undermines the interpretability of the signature when the goal is to infer biological processes causally influenced by $E$.

Another important remark is that, in many practical applications, cycles may be expected, in the sense that several molecular features can influence each other. Consider, for example, two molecular features that influence one another. The resulting cycle $M_1 \leftrightarrows M_2$ can be represented in a ``time-expanded" DAG by assuming that past levels of $M_1$ influence future levels of $M_2$, and vice versa, and possibly also that current levels of $M_1$ influence current levels of $M_2$ (or the reverse, but not both).  The time-expanded version of the DAG is then acyclic, as required, but includes unobserved features, since molecular features are typically measured at one particular time only. Nevertheless, Theorems \ref{thm:no-screening} and \ref{thm:screening} still apply. In the simple example shown in Figure \ref{fig:Dag_time_exp}, assuming only molecular features $M^{(1)}_1, M^{(1)}_2$ and $M^{(1)}_3$ are observed along with exposure $E$,  Theorems \ref{thm:no-screening} and \ref{thm:screening} together with d-separation arguments imply that, asymptotically, the no-screening strategies would select all three features (even though $M^{(1)}_3$ is not causally influenced by $E$), whereas the screening strategy would select only $M^{(1)}_1$ and $M^{(1)}_2$.

\begin{figure}
\centering
\begin{tikzpicture}[>=Latex]

\node (E) {E};

\node[right=1.5cm of E] (M01) {$M^{(0)}_1$};
\node[below=1.5cm of M01] (M02) {$M^{(0)}_2$};
\node[above=1.5cm of M01] (M03) {$M^{(0)}_3$};
\node[right=1.5cm of M01] (M11) {$M^{(1)}_1$};
\node[right=1.5cm of M02] (M12) {$M^{(1)}_2$};
\node[right=1.5cm of M03] (M13) {$M^{(1)}_3$};

\draw[->] (E) -- (M01);
\draw[->] (M01) -- (M11);
\draw[->] (M03) -- (M01);
\draw[->] (M01) -- (M02);
\draw[->] (M01) -- (M12);
\draw[->] (M13) -- (M11);
\draw[->] (M02) -- (M11);
\draw[->] (M02) -- (M12);
\draw[->] (M11) -- (M12);
\end{tikzpicture}\label{fig:Dag_time_exp}

\caption{DAG representing causal relationships among \((E,M_1,M_2, M_3)\) in a simple example where two features ($M_1$ and $M2$)  influence one another. The DAG shows features at times $t_0$ (superscript $(0)$), and $t_1>t_0$ (superscript $(1)$); typically only features at time $t_1$ are observed.}
\end{figure}
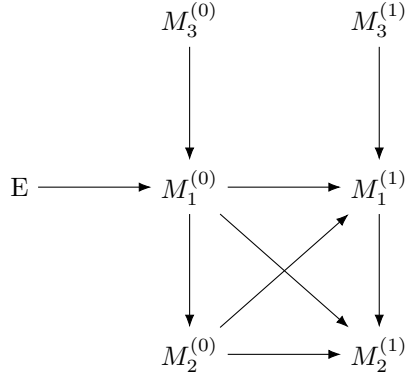

\subsubsection{Proof of Theorem \ref{thm:no-screening}}
Introduce $D'$ any descendant of $D$, or any ancestor of $D$ in $\mathrm{De}(E)$. 
From simple arguments based on $d$-separation, we have
\begin{align*}
E &\not\!\perp\!\!\!\perp M \,\big|\, (D, W) \\
E &\not\!\perp\!\!\!\perp M \,\big|\, (D', W) \\
E &\not\!\perp\!\!\!\perp M \,\big|\, (\bN\setminus\{M\}, W) \quad {\rm if}\ \{D\cup D'\} \cap \bN \neq \emptyset. 
\end{align*}
 
Indeed, since $M$ is a non-descendant of $E$, all paths between \(E\) and \(M\) pass either through a collider such as \(D\) or through one $W_j\in W$. All these paths are d-blocked by $W$, which implies that \(E \perp\!\!\!\perp M\ \,\big|\,  W\) under the Markov assumption. But, conditioning on \(D\) d-opens the paths \(E \to \dots \to D \leftarrow M\) and \(E \to \dots \to D \leftarrow \nu \to M\). Similarly, conditioning on any descendants of $D$ or any of its ancestors in $\mathrm{De}(E)$ d-opens these paths. Under our faithfulness assumption, this induces nonzero conditional dependence. In particular, \(E \not\!\perp\!\!\!\perp M \,\big|\, (D, W)\) and \(E \not\!\perp\!\!\!\perp M \,\big|\, (D', W)\). Adding further variables to the conditioning set does not close the collider-opened path. Hence \(E \not\!\perp\!\!\!\perp M \,\big|\, (\bN\setminus\{M\}, W)\) if $\{D\cup D'\} \cap \bN \neq \emptyset$. 

This result establishes that the coefficient of $M$ in $\alpha^*$ in \eqref{eq:multivar_reg} is non-zero, and any variable-consistent approach used for the derivation of the signature would asymptotically select $M$.  

\subsubsection{Proof of Theorem \ref{thm:screening}}


We first show that for any \(M \in \bM\),  
\begin{equation*}
E  \not\!\perp\!\!\!\perp M \mid W \quad\Longleftrightarrow\quad M \in {\rm De}(E), 
\end{equation*}
which establishes that $(i)$, the screening step asymptotically retains all and only the observed descendants of $E$ and then, $(ii)$, that the features selected with the screening strategy constitute a subset of the observed features in $\mathrm{De}(E)$. 

\begin{proof}
``\(\Rightarrow\)'' If \(E \not\!\perp\!\!\!\perp M \mid W\), the faithfulness assumption implies there is at least one open path between \(E\) and \(M\), that does not go through $W$. Because $E$ is not caused by the molecular features, the only parents of \(E\) are in $W$, and this open path must start with an exiting edge \(E\to\cdot\). Any path that would contain a collider would not be open. As a result, an open path must be a directed chain from \(E \to \dots \to M\), which establishes that  \(M \in {\rm De}(E)\).

``\(\Leftarrow\)'' If \(M\in \mathrm{De}(E)\), there is a directed path \(E\to \cdots \to M\). By faithfulness, \(E \not\!\perp\!\!\!\perp M \mid W\).
\end{proof}

Next, denoting by $\textrm{ch}(E)$ the set of observed features in $\textrm{Ch}(E)$, we show that for any set $S$ of observed features such that ${\rm ch}(E) \subseteq S \subseteq \bN$,
\[ N_j \in {\rm ch}(E) \Rightarrow E \not\!\perp\!\!\!\perp N_j \,\big|\, (S\setminus\{N_j\}, W), \]
which establishes that the parameter $\alpha^*_j$ in \eqref{eq:multivar_reg} is non-zero (under both the no-screening and screening strategies), and then that any variable-consistent approach used for the selection of features would asymptotically select all observed children of $E$. 

\begin{proof}

(2) Because $N_j\in \textrm{ch}(E)$, there is a direct edge \(E\to N_j\) in the DAG. This path is open, and, by faithfulness, remains open  when conditioning on any subset of variables that does not block the edge itself. This establishes that \(E \not\!\perp\!\!\!\perp N_j \mid (S\setminus\{N_j\}, W)\).
\end{proof}


\section{Simulation study} 
\label{sec:Simul}

\subsection{The example of Figure \ref{fig:Dagtoy}} \label{sec:dagtoysim}

We first present results from a simulation study based on the model in Figure \ref{fig:Dagtoy}, to illustrate the theoretical results of Theorems \ref{thm:no-screening} and \ref{thm:screening}.

\subsubsection{Data generation}\label{sec:sim:gendata:toyex}

Variables with no parents in the DAG (namely, $\nu_1, \nu_2, W_1, M_2, M_5, M_7, M_{15}, M_{19}$) were generated independently from a $\mathcal{N}(0,1)$ Gaussian distribution. For each variable with at least one parent, we denote by $\mathrm{Par}$ the set of its (direct) parents in the DAG. These variables (i.e., $E, M_1, M_3, M_4, M_6, M_8, M_9, M_{10}, M_{11}, M_{12}, M_{13}, M_{14}, M_{16}, M_{17}, M_{18}$) were generated independently from a  $\sim {\cal N}( \sum_{X\in {\rm Par}} X/2, 1/2)$ distribution. We considered sample sizes of $ n = 100 \times 5^k$, for $k=1, \dots, 5$ to assess how the no-screening and screening strategies behave as sample size tends to $\infty$. For each sample size, we generated $B = 100$ independent samples of $n$ observations.

\subsubsection{Signature construction and evaluation criteria}\label{sec:Simul:ConstrSign}
Signatures were constructed using a LASSO regression adjusted for $W_1$, as described in \eqref{eq:lasso_unscreen}. The model was implemented using the glmnet R package, with tuning parameter selected using the “one standard-error” rule \cite{friedman_regularization_2010}. Under the no-screening strategy, the LASSO used all available features as input, and we denote the resulting signature by $\hat S_{{\rm no-s}}$. Under the screening strategy, which yielded the signature $\hat S_{{\rm s}}$, it only used features significantly associated with exposure $E$, after adjustment for $W_1$, at 5\% false discovery rate, as described in Section \ref{sec:ScreenNoScreen}.

We compared the performance of the screening and no-screening strategies using two types of criteria. First, we assessed the sets of features $\hat S_{{\rm no-s}}$ and $\hat S_{{\rm s}}$ selected under the no-screening and screening strategies, respectively, and computed the proportion of samples over which \(\hat S_{{\rm no-s}} = S_{{\rm no-s}} = \{M_1, M_2, M_3, M_6, M_7, M_9, M_{11}, M_{14}, M_{15}, M_{18}, M_{19}\}\), and $ \hat S_{{\rm s}} = S_{\rm s} = \{M_1, M_3, M_6, M_{11}, M_{18}\}$. Second, using independent test sets of size $n_{\rm test} = 10,000$, we computed the correlation between the identified molecular signatures and the exposure $E$. 

\subsubsection{Results}

Figure \ref{fig:res:toyEx} summarizes our empirical results. As expected from Theorem \ref{thm:no-screening} and d-separation arguments, the probability that the features selected under the no‑screening strategy match the set $S_{{\rm no-s}} = \{M_1, M_2, M_3, M_6, M_7, M_9, M_{11}, M_{14}, M_{15}, M_{18}, M_{19}\}$ converged to 1 as $n\rightarrow \infty$. Similarly, our empirical results confirmed that the probability that the features selected under the screening strategy match the set $S_{{\rm s}} = \{M_1, M_3, M_6, M_{11}, M_{18}\}$ converges to 1 as $n\rightarrow \infty$, as expected from Theorem \ref{thm:screening} and d-separation arguments. Our results further showed that the signature $S_{{\rm no-s}}$ exhibited a stronger correlation with $E$ on independent test sets than $S_{{\rm s}}$. Finally, the bottom panel of Figure \ref{fig:res:toyEx} shows that the signatures generally include all features in $S_{{\rm no-s}}$ (for the no-screening strategy)  or $S_{{\rm s}}$ (for the screening strategy), but they also often included additional features, especially for small to moderate sample sizes, which is a well‑known behavior of the LASSO \cite{buhlmann_statistics_2011}.

\begin{figure} 
{\hspace{-50pt}}\includegraphics[scale=0.35]{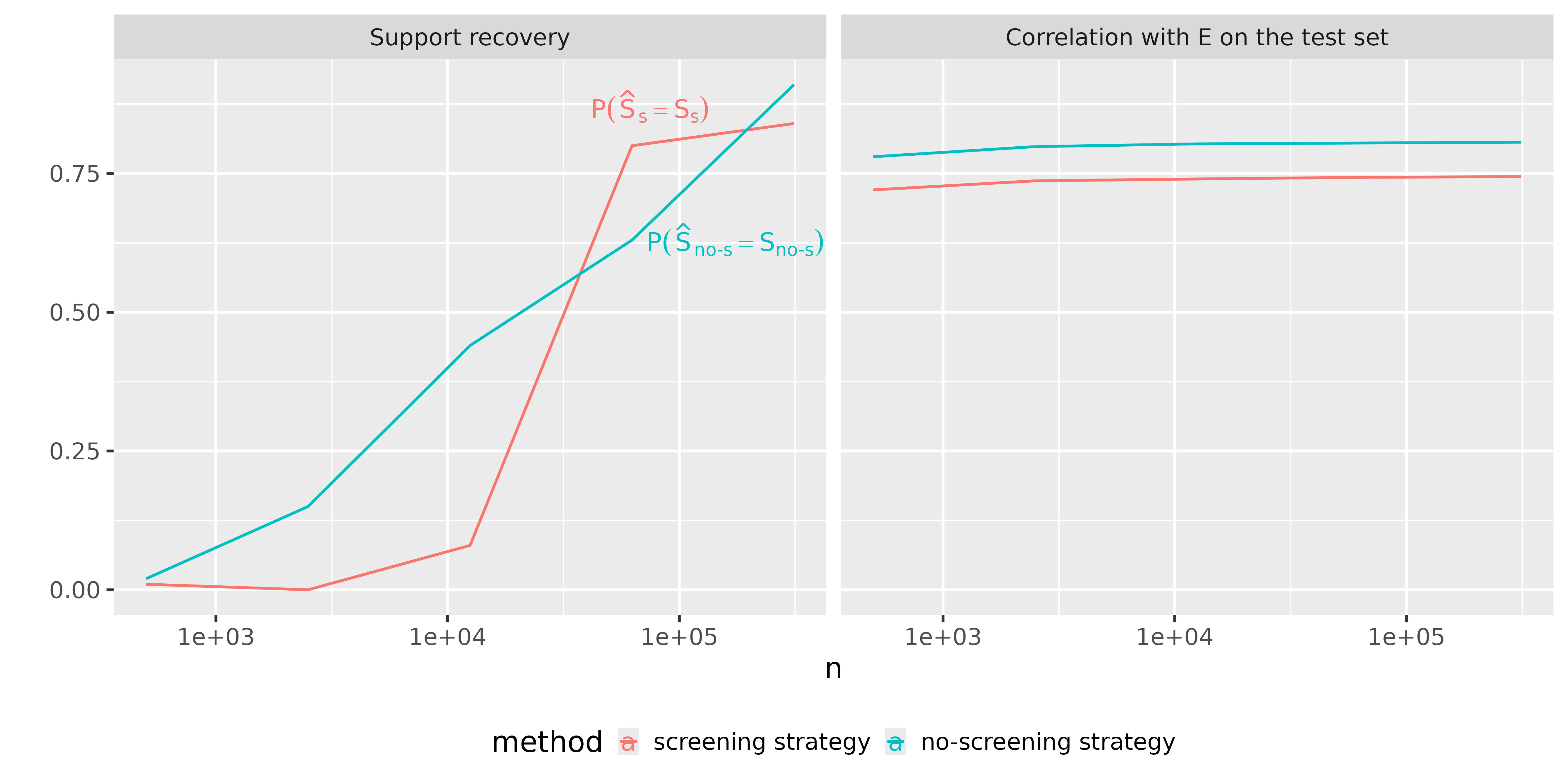}\\\vskip5pt
{\hspace{-50pt}}\includegraphics[scale=0.35]{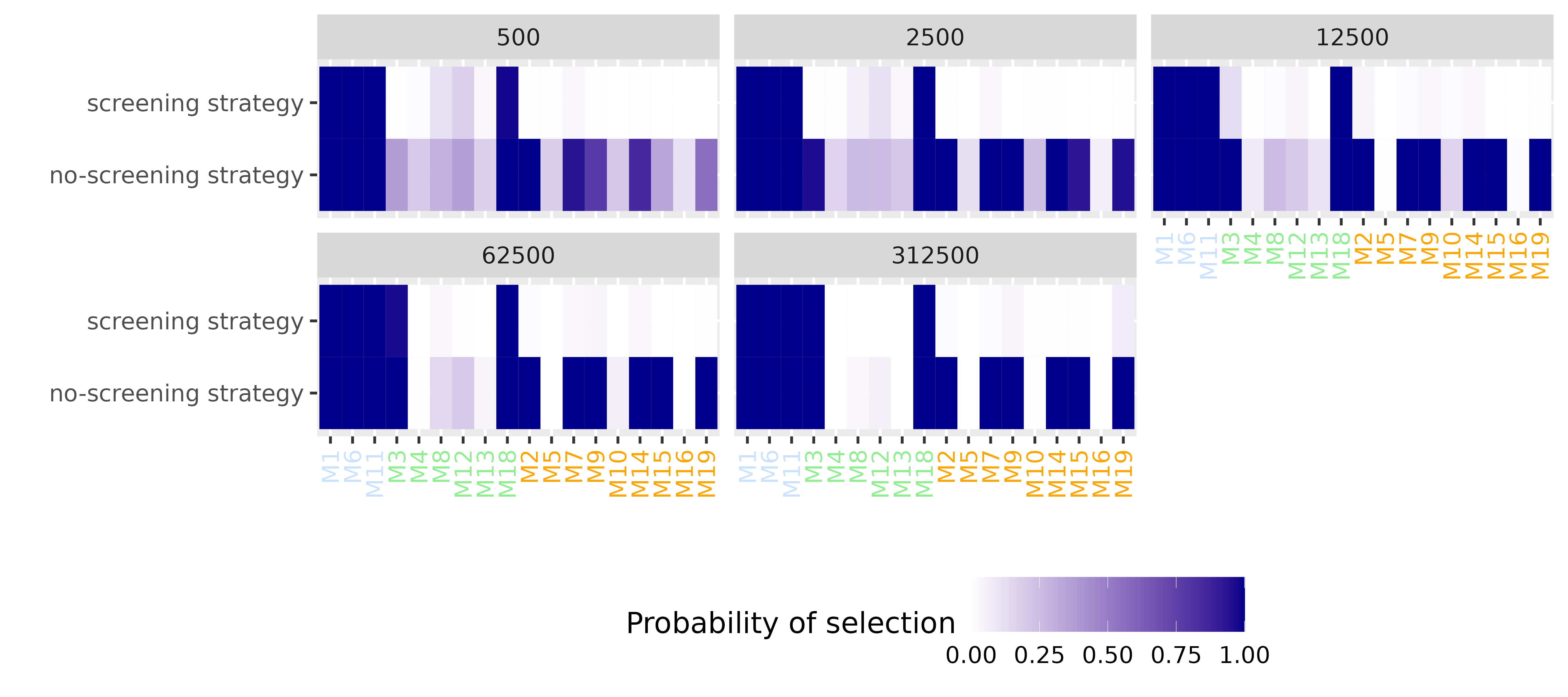}

\caption{Results of the simulation study for the toy example shown in Figure \ref{fig:Dagtoy}. The top-left panel displays, as a function of sample size, the proportion of replicas (out of 100) for which the set of features selected by the no‑screening strategy exactly matches $S_{{\rm no-s}} = \{M_1, M_2, M_3, M_6, M_7, M_9, M_{11}, M_{14}, M_{15}, M_{18}, M_{19}\}$ (cyan), and the proportion for which the features selected by the screening strategy exactly match $S_{\rm s} = \{M_1, M_3, M_6, M_{11}, M_{18}\}$ (red). The top-right panel reports the correlation (averaged over 100 replicas) between exposure $E$ and the resulting signature, evaluated on an independent test dataset of size $n_{\mathrm{test}}=10,000$, for the no‑screening (cyan) and screening (red) strategies. The bottom panel shows, for each feature, its selection probability across the 100 replicas at each sample size. For readability, features are ordered as children of $E$ (light blue), other descendants of $E$ (light green), and non‑descendants (orange).}
\label{fig:res:toyEx}
\end{figure}

\subsection{Additional scenarios}

We now consider additional scenarios to illustrate the finite‑sample performance of the no‑screening and screening strategies in settings with high‑dimensional molecular data, which are common in practice. We examine the two scenarios of Figure \ref{fig:dags} and a third, extended scenario corresponding to Figure \ref{fig:Dagtoy}. Scenario $(i)$ corresponds to the case where no collider is present, and where features are either descendants of $E$ or not associated with $E$ at all (neither marginally nor upon conditioning on other features). Scenario $(ii)$ corresponds to a case with colliders, where the DAG contains paths of the form $E \rightarrow M_k \leftarrow \nu \rightarrow M_\ell$. For simplicity, no confounders $W$ are included and all features are observed $(p=p_0)$. 

Scenario $(iii)$ is an extension of the structure presented in Section \ref{sec:dagtoysim} and Figure \ref{fig:Dagtoy}, designed to evaluate the strategies in a high-dimensional setting that may more closely mimic real-world omics data. While the initial toy example focused on a single set of 19 molecular features, this scenario scales the feature space to $p\in(1045, 2090)$ features by repeating the 19-feature block 55 or 110 times. In this setup, the exposure $E$ and the shared confounder $W_1$ are generated once and serve as the common ancestors for all 55 blocks of features. This scenario specifically tests the robustness of the screening and no-screening strategies when faced with a high-dimensional matrix containing a large number of potential colliders and redundant causal paths.

From Theorems \ref{thm:no-screening} and \ref{thm:screening}, together with d-separation, it follows that features selected using the screening strategy asymptotically form the set $S_{\rm s} = \{M_1, \dots, M_{p_{\rm child}}\}$ under scenarios $(i)$ and $(ii)$. In contrast, under the no-screening strategy, the selected features asymptotically form the set $S_{\rm no-s} = S_{\rm s} = \{M_1, \dots, M_{p_{\rm child}}\}$ under Scenario $(i)$, while they form the set of features $S_{\rm no-s} = \{M_1, \dots, M_{p_8}\} $ under Scenario $(ii)$ (see Figure \ref{fig:dags} for the definition of $p_8$). Under scenario $(iii)$, features selected using the screening strategy asymptotically form the generalized set $S_{\rm s} =\{M_{(b-1)\times19+p)}:b\in\{1,\dots,B\}, p\in \{1,3,6,11,18\}\}$, while under the no-screening strategy, the selected features asymptotically form the set $S_{\rm no-s} =\{M_{(b-1)\times19+p)}:b\in\{1,\dots,B\}, p\in \{1,2,3,6,7,9,11,14,15,18,19\}\}$ for $B \in \{55,110\}$ blocks.

\usetikzlibrary{arrows.meta, positioning}
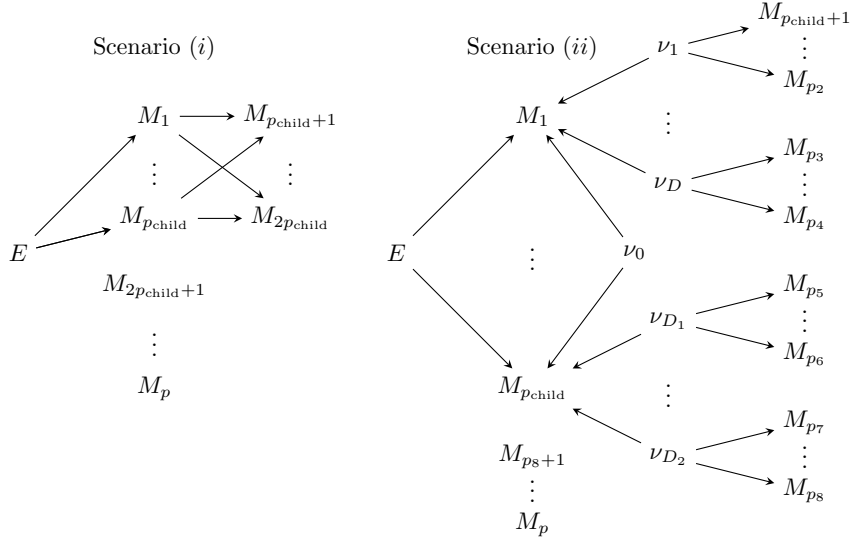
\begin{figure}
    \centering
    \begin{tikzpicture}
    \tikzset{
        > = stealth,
        every path/.append style = {arrows = ->},
        node distance = 1cm and 1cm
    }
        \begin{scope}[scale=0.9, transform shape]
               \node (x1) at (0,0) {$E$};
            \node (m1) at (2,2) {$M_1$};
            \node at (2,1.25) {$\vdots$};
            \node (mprel) at (2,0.5) {$M_{p_{\rm child}}$};

            \node (mprelplus1) at (4,2) {$M_{p_{\rm child}+1}$};
            \node at (4,1.25) {$\vdots$};
            \node (m2prel) at (4,0.5) {$M_{2p_{\rm child}}$};    
            
            \node (m2prelplus1) at (2,-0.5) {$M_{2p_{\rm child}+1}$};
            \node at (2,-1.25) {$\vdots$};
            \node (mp) at (2,-2) {$M_{p}$};
            \path (x1) edge (m1);
            \path (x1) edge (mprel);
            \path (m1) edge (mprelplus1);
            \path (m1) edge (m2prel);
            \path (mprel) edge (mprelplus1);
            \path (mprel) edge (m2prel);
            \path (x1) edge (mprel);
            
            \node at (2,3) {Scenario $(i)$};
        \end{scope}

        \begin{scope}[xshift=5cm, scale=0.9, transform shape]
            \node (x1) at (0,0) {$E$};
            \node (m1) at (2,2) {$M_1$};
            \node at (2,0) {$\vdots$};
            \node (mpc) at (2,-2) {$M_{p_{\rm child}}$};

            \node (n0) at (3.5,0) {$\nu_0$};
            \node (n1) at (4,3) {$\nu_1$};
            \node at (4,2) {$\vdots$};
            \node (nB) at (4,1) {$\nu_D$};
            \node (n2) at (4,-1) {$\nu_{D_1}$}; 
            \node at (4,-2) {$\vdots$};
            \node (n3) at (4,-3) {$\nu_{D_2}$}; 
            \node (mp1) at (6,3.5) {$M_{p_{\rm child}+1}$};
            \node at (6,3.1) {$\vdots$};
            \node (mp2) at (6,2.5) {$M_{p_2}$};
            \node (mp3) at (6,1.5) {$M_{p_3}$};
            \node at (6,1.1) {$\vdots$};
            \node (mp4) at (6,0.5) {$M_{p_4}$};
            \node (mp5) at (6,-0.5) {$M_{p_5}$}; 
            \node at (6,-0.9) {$\vdots$};
            \node (mp6) at (6,-1.5) {$M_{p_6}$}; 
            \node (mp7) at (6,-2.5) {$M_{p_7}$};
            \node at (6,-2.9) {$\vdots$};
            \node (mp8) at (6,-3.5) {$M_{p_8}$};

            \node (mp9) at (2,-3) {$M_{p_8 + 1}$};
            \node at (2,-3.4) {$\vdots$};
            \node (mp) at (2,-4) {$M_{p}$};
            
            \path (x1) edge (m1);
            \path (x1) edge (mpc);

            \path (n0) edge (m1);
            \path (n0) edge (mpc);
            
            \path (n1) edge (m1);
            \path (n1) edge (mp1);
            \path (n1) edge (mp2);
            \path (nB) edge (m1);
            \path (nB) edge (mp3);
            \path (nB) edge (mp4);
            \path (n2) edge (mpc);
            \path (n2) edge (mp5);
            \path (n2) edge (mp6);
            \path (n3) edge (mpc);
            \path (n3) edge (mp8);
            \path (n3) edge (mp7);
            \node at (2,3) {Scenario $(ii)$};
        \end{scope}

    \end{tikzpicture}

    \caption{Directed acyclic graphs illustrating the additional simulation scenarios considered to evaluate the performance of feature selection strategies in the presence of different dependence structures among features: $(i)$ a scenario without colliders, where some features are descendants of $E$ and others are not associated with $E$ at all (neither marginally nor upon conditioning on other features); $(ii)$ A complex dependency structure where exposure-related features (``children", i.e. $M_1, \dots, M_{p_{\rm child}}$) and non-descendant features (``step-brothers" $M_{p_{\rm child}+1}, \dots, M_{p_{8}}$) share common latent causes (``mothers" $\nu_1, \ldots, \nu_{D_2}$). This design introduces collider paths ($E$→Child←Mother→Step-brother).}
    \label{fig:dags}
\end{figure}

\begin{figure}
\centering
\includegraphics[width=0.9\textwidth]{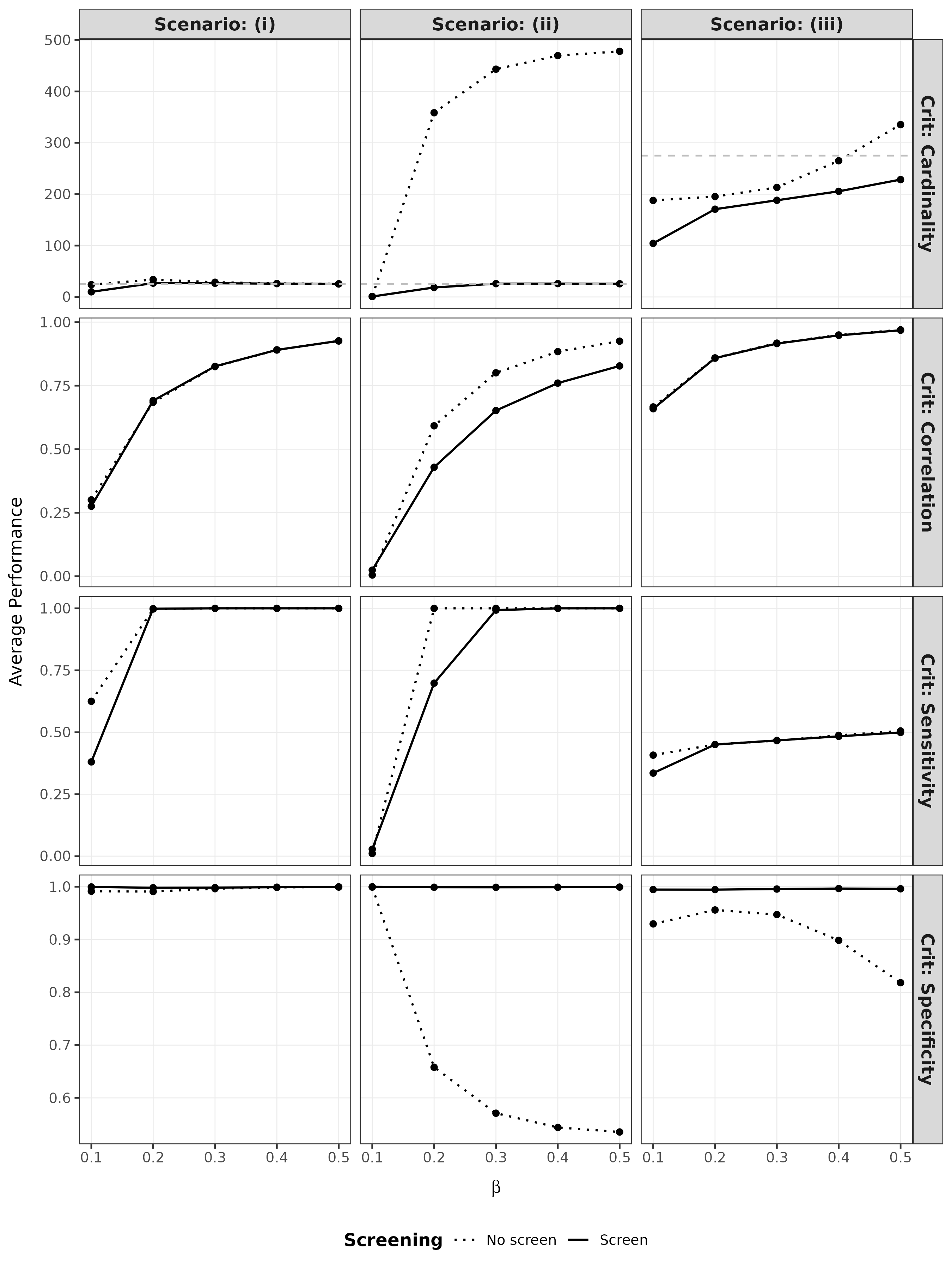}
\caption{Simulation results comparing feature selection strategies across three scenarios ($n=1000$). The grey dashed line in the first row represents the true number of (observed) descendants of the exposure.
\textbf{Scenario $(i)$}: $p = 1000$; The exposure was directly associated with the first 25 features, while the second 25 features were associated with the first 25 features (linear regression coefficients uniformly distributed between 0.005-0.01). \textbf{Scenario $(ii)$}: $p = 1000$; The exposure was associated with 25 features, with $D=5$ mothers and $d=5$ stepbrothers per mother. \textbf{Scenario $(iii)$}: $p = 1054$; The features follow 55 repeated blocks of the DAG in Figure \ref{fig:Dagtoy}, including 495 observed causal descendants (9 per block), multiple non-descendants, and an observed shared confounder $W_1$. Under the screening strategy, 275 features would be asymptotically selected.}\label{fig2}
\end{figure}

\subsubsection{Data generation}

For scenarios $(i)$ and $(ii)$, we simulated data for $n=1,000$ individuals with information on exposure $E$ and molecular features $M_1, \ldots, M_{p}$, with $p=1,000$. The exposure $E$ followed a ${\cal N}(0,1)$ Gaussian distribution throughout all simulations. Denoting by $p_{\rm child}$ the number of features influenced by $E$, and by $\beta$ the effect size, features were generated as follows. We varied the number $p_{\rm child}$ of features directly influenced by $E$ in $\{5, 25, 125\}$. 

Under scenario $(i)$, we set $M_j = \beta  E  + \epsilon_j$, for $j=1,\ldots, p_{\rm child}$; $M_j = \sum_{k=1}^{p_{\rm child}} u_{k,j}M_k + \epsilon_j$ for $j=p_{\rm child}+1, \ldots, 2p_{\rm child} $; and $M_j = \epsilon_j$ for $j= 2p_{\rm child}+1, \ldots, 1000$. Parameters $u_{k,j}$, linking features $M_1, \ldots, M_{p_{\rm child}}$ to features $M_{p_{\rm child}+1}, \ldots, M_{2p_{\rm child}}$, were drawn independently from a Uniform [0.005, 0.01]-distribution, and the $\epsilon_j$'s were generated independently from a ${\cal N}(0,1)$ Gaussian distribution. 

Under Scenario $(ii)$, let $D$ denote the number of ``mothers'' per child (e.g., $\nu_1, \dots, \nu_D$ for $M_1$ in Figure \ref{fig:dags}), which link a given child to its ``step-brothers'' (e.g., features $M_{p_{\rm child}+1}, \dots, M_{p_2}$ for $M_1$ in Figure \ref{fig:dags}). For $j\in \{1, \dots, p_{\rm child}\}$ 
let ${\cal M}_j = \{k\in \mathbb{N}:(j-1)D<k\leq jD\}$ such that $(\nu_k)_{k\in {\cal M}_j}$ are the mothers of $M_j$.
Finally let $d$ denote the number of step-brothers per mother. We first generated latent variables $\nu_0, \nu_1, \ldots, \nu_{Dp_{\rm child}}$ under independent ${\cal N}(0, 1)$ Gaussian distributions. Then, we set $M_j = \beta  E + \delta_g \nu_0 + \sum_{k \in {\cal M}_j}  \delta \nu_{k} + \epsilon_j$, for $j=1, \dots, p_{\rm child}$; $M_j = \delta \nu_{k_j} + \epsilon_j$ for $j=p_{\rm child}+1, \dots, p_{\rm child}(1 + Dd)$, with 
$k_j = \lfloor\frac{j-p_{\rm child}-1}{d} + 1\rfloor$.
and $M_j = \epsilon_j$ for $j= p_{\rm child}(1 + Dd) + 1, \ldots, 1000$. The $\epsilon_j$'s were generated independently from ${\cal N}(0,0.25)$ Gaussian distributions for $j=1, \dots, p_{\rm child}(1 + Dd)$ and from ${\cal N}(0,1)$ Gaussian distributions for $j=p_{\rm child}(1 + Dd) + 1, \ldots, 1000$. We set $\delta = 0.5$ in all experiments, and varied the number $D$ of mothers per child and the number $d$ of step-brothers per mother, as well as parameter $\delta_g$.

For scenario $(iii)$, we simulated data for $n=1000$ and $p=1045$, as well as for for $ n = 100 \times 5^k$, for $k\in\{1,2, 4\}$,  individuals and $p \in \{1045, 2090\}$, based on the DAG structure described in section \ref{sec:dagtoysim}. To adapt this structure for a high-dimensional setting, the shared confounder $W_1$ and exposure $E$ were generated once for each individual, such that $W_1 \sim {\cal N}(0,1)$ , and $E \sim {\cal N}( \beta W_1, 0.5)$. The 19-feature molecular blocks $\{M_{19(b-1) + 1}, \dots, M_{19b}\}$, for $b=1,\dots B$ (with $B=55$ or $B=110$) were then generated repeatedly, following the same approach as for features $M_1, \dots M_{19}$ in Section \ref{sec:sim:gendata:toyex}. 

We varied the signal strength $\beta$ in $\{0.1, 0.2, 0.3, 0.4, 0.5\}$ for all scenarios. 

\subsubsection{Signature construction and evaluation criteria}\label{perf_metrics}

Signatures were constructed under the no-screening and screening strategies, as described in Section \ref{sec:Simul:ConstrSign} above. Performance of the two strategies was assessed using four criteria: (1) the cardinality of the identified signature $\hat S$ (i.e., the number of features in $\hat S$); (2) the correlation between $\hat S$ and exposure $E$, assessed on an independent test set of size 10,000; (3) the ``sensitivity", representing the ability of the signature to select the asymptotic set $S_{\rm s}$ ($\P(M_j \in  \hat S | M_j \in S_{\rm s})$); and (4) the ``specificity'', here representing the ability to exclude non-descendants of $E$ ($\P(M_j \notin \hat S | M_j \notin {\rm De}(E))$).

\subsubsection{Main results} \label{sec:Results}
Figure \ref{fig2} presents the average performance of the evaluation criteria over 100 replications for each scenario and $\beta$ value. Results are shown for the case where $p_{\rm child}=25$ for scenarios $(i)$ and $(ii)$, and $n = 1000, p=1045$ for scenario $(iii)$. Additional results for other values of $p_{\rm child} \in \{5,25,125\}$ for scenarios $(i)$ and $(ii)$ and other values of $n$ and $p$ for scenario $(iii)$ are presented in the Supplementary Material. Overall, the findings were consistent with our theoretical expectations and highlight a trade-off induced by the screening step.
In scenario $(i)$, where collider bias is absent, screening has little to no effect (Figure \ref{fig2}). 
Under scenarios $(ii)$ and $(iii)$, where collider bias is present, the no-screening strategy produced signatures with low specificity, and whose cardinality often substantially exceeds the number of descendants. The screening step effectively mitigated this effect, albeit at the potential cost of reduced correlation (under scenario $(ii)$) and reduced sensitivity (for $\beta=0.2$ under scenario $(ii)$ and $\beta=0.1$ under scenario $(iii)$).  

Additional results presented in the Supplementary Material confirmed the patterns observed on the main results. In scenario $(i)$ (Figure \ref{supp_fig_scen1}), there is little effect of screening across all values of $p_{\rm child}$. In scenario $(ii)$, the screening step improved specificity, but at the possible cost of reduced sensitivity (for low signal strengths) and reduced correlation with the exposure, across most tested values of $p_{\rm child}$, $D$, $d$ and $\delta_g$ (Figure \ref{fig2}). Notably, for some parameter configurations (especially when $\delta_g = 0.2$), collider bias was limited and the no-screening strategy already produced specific signatures. In scenario $(iii)$, the no-screening strategy produced signatures with low specificity, especially for large enough $n/p$ ratio and/or signal strength, while the screening strategy corrected for this limitation at the cost of reduced sensitivity and correlation with the exposure (e.g., for low signal strength in the $n=500$ setting) (Figure \ref{fig2}.

\section{Discussion} \label{sec:Disc}

Signatures of lifestyle exposures are often constructed using multivariate linear models, yet the application of an initial univariate screening step is applied inconsistently across studies and without formal guidance on their causal implications. In this article, we develop both formal and empirical arguments to evaluate the respective merits of the no-screening and screening strategies  used in the literature for the construction of molecular signatures of lifestyle exposures. The first approach uses all available features, whereas the second applies an initial ``univariate'' screening step to focus the construction of the signatures on features ``marginally'' associated with the exposure (considering one feature at a time, but adjusting for possible confounders).  

Using d-separation arguments, we showed that a particular form of collider bias may arise when constructing molecular signatures of lifestyle exposures. When present, this collider bias leads the no‑screening strategy to include non‑causal features in the signature, whereas this does not occur under the screening strategy. Throughout this article, we took the example of lifestyle exposures, but importantly, this collider bias may arise when constructing molecular signatures of any other phenotype that causally influences the molecular features. Conversely, this collider bias does not arise if the studied phenotype is causally influenced by the molecular features. For example, molecular signatures have been developed as predictive biomarkers of disease risk, such as for cancer \cite{hu_metabolomic_2024,wilechansky_prediagnostic_2025, li_prediagnostic_2024}. The molecular features used to derive disease-risk signatures are typically pre-diagnostic and considered causes (or proxies for causes) of the disease being studied, rather than their consequences. In such settings, the collider bias we describe in this work would not occur, and the screening step would not be needed to avoid the inclusion of non-causal features in the signature.  

Focusing on molecular signatures of lifestyle exposures (or, more generally, phenotypes that causally influence molecular features), our results indicate that the screening step consistently improves specificity. In some settings, however, this gain comes at the cost of reduced sensitivity and weaker associations with the exposure.

In our experiments, the observed reduction in sensitivity was usually not driven by failures of the initial screening stage to retain features causally affected by the exposure, but rather by how screening constrains the subsequent regression used for selection. The coefficient vector $\boldsymbol\alpha$ of the regression of $E$ on the features depends on the causal effects of $E$ on the features, but also on the covariance matrix of the features. When screening is omitted and the regression is performed on all features, the inclusion of many correlated but non‑causal features can inflate the magnitude of $\boldsymbol\alpha$ for the causal features, thereby increasing their probability of being selected in the signature. By contrast, screening limits this correlation‑induced amplification, which may reduce sensitivity when causal effects are moderate or weak (Figures \ref{supp_fig_scen2_dg0} and \ref{fig:supp_scen3_allcrit}); a detailed explanation is provided in the Supplementary Material \ref{appendix:dependence_sensitvity}.



The weaker association with the exposure for signatures derived through the screening strategy reflects the general fact that non-causal features can, in some settings, improve predictive performance. In our experiments, this attenuation was primarily observed under some settings of scenario $(ii)$, regardless of the signal strength, as well as at low signal strength under scenario $(iii)$ (e.g., $\beta = 0.1$ in Figure \ref{fig:supp_scen3_allcrit}). 

Overall, the choice to apply univariate screening depends on whether the objective is etiologic interpretation or prediction. If the goal is to capture metabolic responses to the exposure, screening is advantageous as it enhances biological relevance, and the associated loss in sensitivity is not systematic. In contrast, when the signature is intended to serve as a biomarker, the trade-off becomes less straightforward, as predictive performance may take precedence over specificity. On the one hand, clinical biomarkers are typically expected to be specific to the exposure, supporting the use of screening to exclude non-causal features \cite{cuparencu_towards_2024}. On the other hand, non-causal features may improve predictive accuracy, at least within the population under study. Although it has been argued that such signatures may generalize poorly to external populations \cite{s_selecting_2023}, this concern has recently been challenged in the machine learning literature \cite{nastl_causal_2024}. Consequently, when the objective is prediction, the screening step may not be necessary.

Similar trade-offs may arise when deciding whether to adjust for correlated exposures or other confounders during signature construction. For example, in the toy example of Figure \ref{fig:Dagtoy}, feature $M_{16}$ would be selected in the signature if confounder $W_1$ was not adjusted for. Published studies on molecular signatures differed in whether they adjusted for potential confounders \cite{zhu_proteomic_2025,wan_plasma_2025,smith_healthy_2022}. An example illustrating the importance of adjusting for confounders comes from a study of metabolic signatures of alcohol and smoking, two highly correlated exposures \cite{assi_are_2018}. Ignoring adjustment for the alternate exposure, the respective signatures shared four of seven features, while no features were shared after adjustment for the alternate exposure. This highlights how confounding can obscure the specificity of signatures. 

Several limitations should be considered when interpreting these findings. Our simulations relied on simple causal structures, linear relationships, and lasso-based signature construction. Real molecular systems may exhibit nonlinear effects, feedback loops, heterogeneous correlation structures, and substantially higher dimensionality. An additional consideration not explored here is measurement error in the exposure, which is common in epidemiologic studies of lifestyle factors. Such error may reduce sensitivity by attenuating true associations below the screening threshold, and may also reduce specificity through residual confounding.


To recap, we demonstrated that the application of univariate screening should be carefully considered when constructing molecular signatures of a phenotype. When the studied phenotype can be seen as a cause of the available molecular features, our results indicate that univariate screening is recommended when investigating causal mechanisms of exposures through molecular signatures. While it could be preferable to ignore this step when the signature is intended to be used as a predictive biomarker, our simulations show that univariate screening may retain similar performance to its omission. Our work provides a foundation for methodological choices in the construction of molecular signatures, particularly as they continue to be used in epidemiology.

\bigskip

\bibliography{Signatures_WuViallon}

\newpage
\begin{appendices}

\section{Additional results}\label{secA1}


\subsection{Results of additional figures}
Figure \ref{supp_fig_scen1} presents results for scenario $(i)$ varying the number of children and the strength of the signal. The findings of this scenario are consistent with our theoretical expectations, with little difference in cardinality, correlation, sensitivity, and specificity across the screening and no-screening strategies. 

Figure \ref{supp_fig_scen2_dg0} presents results for scenario $(ii)$ with $\delta _g =0$, varying the number of children and corresponding mothers and step-brothers, as well as the strength of the signal. The specificity was 1 for the screening strategy across all scenario, while the no-screening strategy generally achieved a lower specificity, but a larger association with the exposure and larger sensitivity (especially for lower signals). The exception is for a higher number of mothers than step-brothers per child when $p_{\rm child}=5$ and $\beta \leq 0.4$, where there was little difference in all four metrics.

Figure \ref{fig:supp_fig_scen2_dg0.2} presents results for scenario $(ii)$ with $\delta _g =0.2$, varying the number of children and corresponding mothers and step-brothers, as well as the strength of the signal. Across most scenarios, the screening and no-screening strategies performed similarly in terms of cardinality, association with the exposure, sensitivity, and specificity. For some scenario (e.g., $(p_{\rm child}=5, D=7, d=25)$ and $p_{\rm child}=25, D=5, d=7$), the no-screening strategy yielded a decreased specificity (and increased cardinality) for larger signal strengths. The no-screening strategy also led to a larger association with the exposure for larger signal strengths in the $(p_{\rm child}=5, D=7, d=25)$ setting.

Figure \ref{fig:supp_scen3_allcrit} presents the performance under scenario $(iii)$  (repeated blocks of 19 features as in Figure \ref{fig:Dagtoy}) across varying sample sizes, numbers of features and signal strengths. For all sample sizes and number of features, the screening and no-screening strategy achieved similar sensitivity and association with the exposure, except for the lowest signal strength ($beta=0.1$), where the no-screening strategy performed slightly better. As sample size and $\beta$ increase, the no-screening strategy exhibits low specificity (and large cardinality), whereas the screening strategy maintained a perfect specificity. As a complement, Figure \ref{fig:supp_scen3_selfreq} presents the proportions of  selected children, descendants, and non-descendants under the same scenarios as in Figure \ref{fig:Dagtoy}. Selection frequencies confirm that the no-screening strategy increasingly includes non-descendant features as sample size and $\beta$ grow, reflecting collider-induced associations. By contrast, the screening strategy concentrated selection on children and descendants of the exposure, resulting in substantially higher specificity.

\begin{figure}[ht]
\centering
\includegraphics[width=0.9\textwidth]{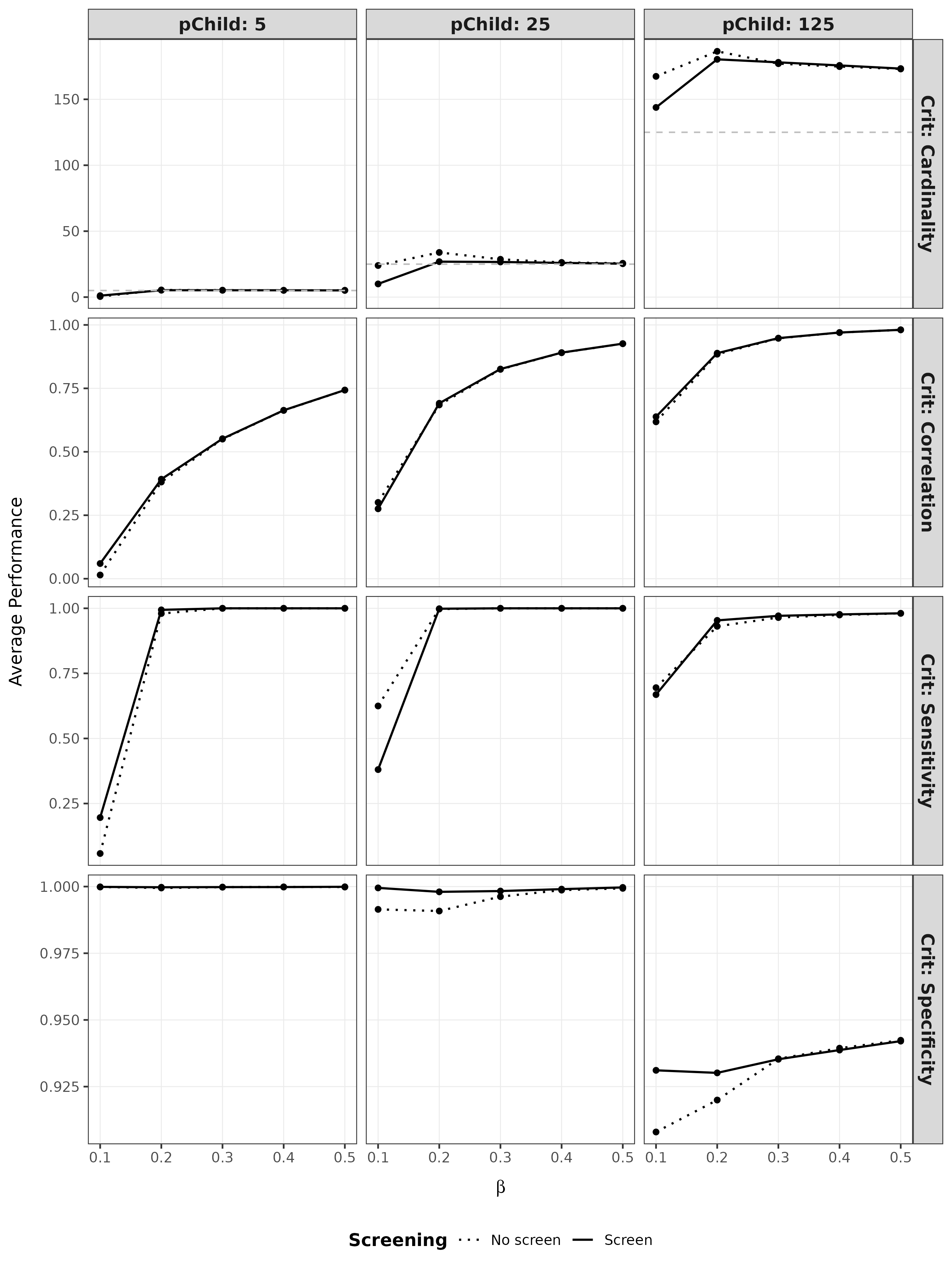}
\caption{Simulation results comparing feature selection strategies across $p_{child}\in 5,25,125$ for scenario $(i)$. Each simulation included 1,000 individuals and 1,000 molecular features, with one exposure variable, and $p_{child}$ related features. Performance was evaluated using four metrics: (row 1) correlation between the exposure and the selected feature signature, (row 2) overall sensitivity, (row 3) sensitivity to exposure-related latent variables, and (row 4) specificity. The grey dashed line in the first row represents the true number of related features.}\label{supp_fig_scen1}
\end{figure}

\begin{figure}[ht]
\centering
\includegraphics[width=0.9\textwidth]{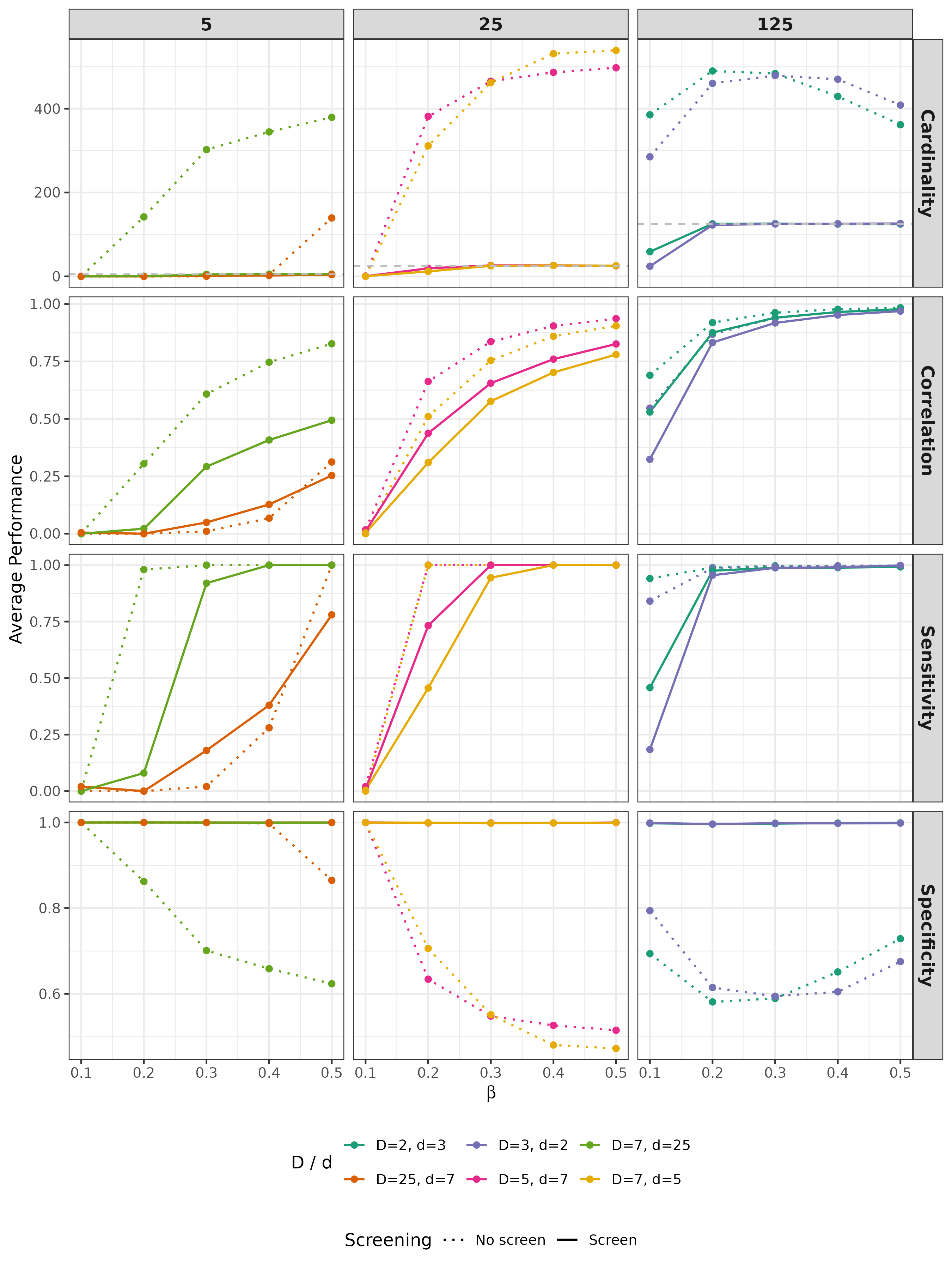}
\caption{Simulation results comparing feature selection strategies across $p_{child}\in5,25,125$ and $\delta _g =0$ for scenario $(ii)$. Each simulation included 1,000 individuals and 1,000 molecular features, with one exposure variable, and $p_{child}$ related features. Performance was evaluated using four metrics: (row 1) correlation between the exposure and the selected feature signature, (row 2) overall sensitivity, (row 3) sensitivity to exposure-related latent variables, and (row 4) specificity. The grey dashed line in the first row represents the true number of related features.}\label{supp_fig_scen2_dg0}
\end{figure}

\subsection{Dependence of sensitivity on the covariance structure of the features} \label{appendix:dependence_sensitvity}
In some settings, we observed that the screening strategy could lead to a reduction in sensitivity. This phenomenon, which was particularly apparent in scenario $(ii)$ with $p_{\rm child} = 5$, $D=7$ and $d=25$ in Figure \ref{supp_fig_scen2_dg0}, can be explained as follows. 

Under the linear structural equation models considered in our experiments (and ignoring intercept terms for simplicity), it follows from simple algebra that the coefficient vector of the regression of $E$ on $\bM_{S}$, for any subset of features $\bM_{S} \subseteq \bM$, is given by $$\boldsymbol\alpha_S = \mathrm{Var}(E)\,\Sigma_S^{-1}\boldsymbol\beta_S,$$ where $\Sigma_S$ denotes the covariance matrix of $\bM_{S}$, and $\boldsymbol{\beta}_S$ denotes the vector of causal effects of $E$ on $\bM_{S}$. As a consequence, the magnitude of the components of $\boldsymbol\alpha_S$, which directly drives lasso selection when constructing the signature, depends critically on the covariance structure of $\bM_{S}$.

If screening performs perfectly, $\bM_{S}$ comprises all and only the causal features (i.e., children and descendants of $E$), and the resulting coefficients $\boldsymbol\alpha_S$ can be viewed as ``oracle'' values for the regression‑based procedure. By contrast, when screening is omitted and the regression is performed on a much larger set of features, the inclusion of many correlated but non‑causal features alters $\Sigma_S$ in a way that can substantially increase the magnitude of $\boldsymbol\alpha_S$ for the causal features, even though the underlying causal effects $\boldsymbol\beta_S$ remain unchanged. This correlation‑induced amplification increases the probability that causal features are selected by the lasso. Although the lasso typically selects only a subset of the non‑causal features, this partial inclusion is sufficient to induce the amplification effect. Consequently, screening may appear to reduce sensitivity, not because it removes causal features, but because it prevents an artificial inflation of the regression coefficients used for selection.

For illustration, consider the setting with $\delta_g=0, p_{\rm child} = 5$, $D=7$, $d=25$ and $\beta = 0.2$  under scenario $(ii)$. When $\bM_S$ comprises all children of $E$ and no other features, the components of $\boldsymbol{\alpha}_S$ are all approximately $0.0125$ (based on estimation using $10^5$ observations). When $\bM_S$ comprises all children as well as one additional step-brother per child, the components of $\boldsymbol{\alpha}_S$ corresponding to the children increase to $0.06$ (while those corresponding to non-descendants are around $ -0.04$). Finally, when $\bM_S$ comprise all 1,000 features,  the components of $\boldsymbol{\alpha}$ corresponding to the children of $E$ reach $0.95$, while those corresponding to non-descendants are approximately $-0.03$. 
This illustrates that the no-screening strategy, while detrimental for specificity, may induce a geometry of the regression problem that increases the probability of including causal features in the signature under certain correlation structures.

\begin{figure}[ht]
\centering
\includegraphics[width=0.9\textwidth]{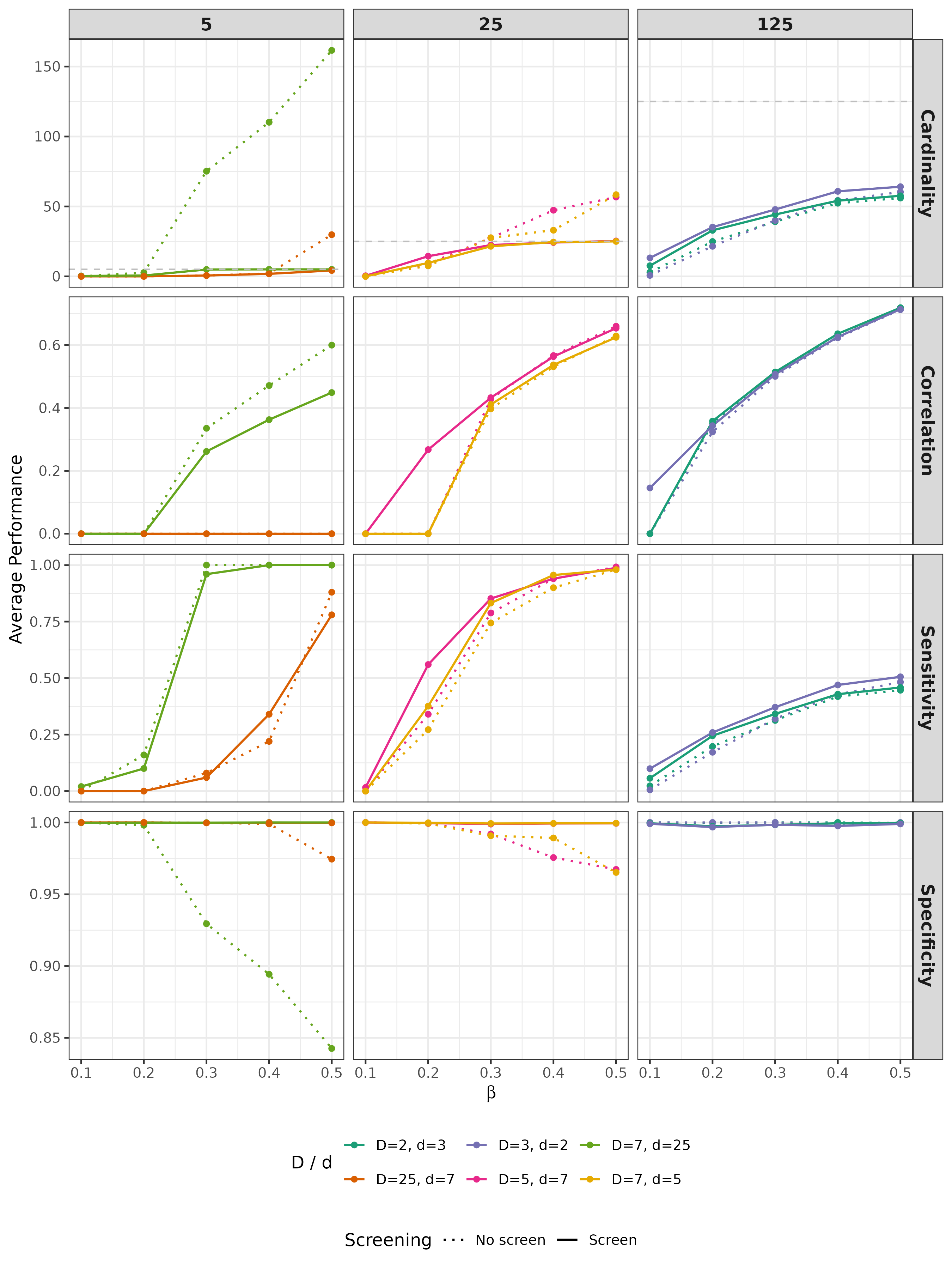}
\caption{Simulation results comparing feature selection strategies across $p_{child}\in5,25,125$ and $\delta _g =0.2$ for scenario $(ii)$. Each simulation included 1,000 individuals and 1,000 molecular features, with one exposure variable, and $p_{child}$ related features. Performance was evaluated using four metrics: (row 1) correlation between the exposure and the selected feature signature, (row 2) overall sensitivity, (row 3) sensitivity to exposure-related latent variables, and (row 4) specificity. The grey dashed line in the first row represents the true number of related features.} \label{fig:supp_fig_scen2_dg0.2}
\end{figure}

\begin{figure}[ht]
\centering
\includegraphics[width=0.9\textwidth]{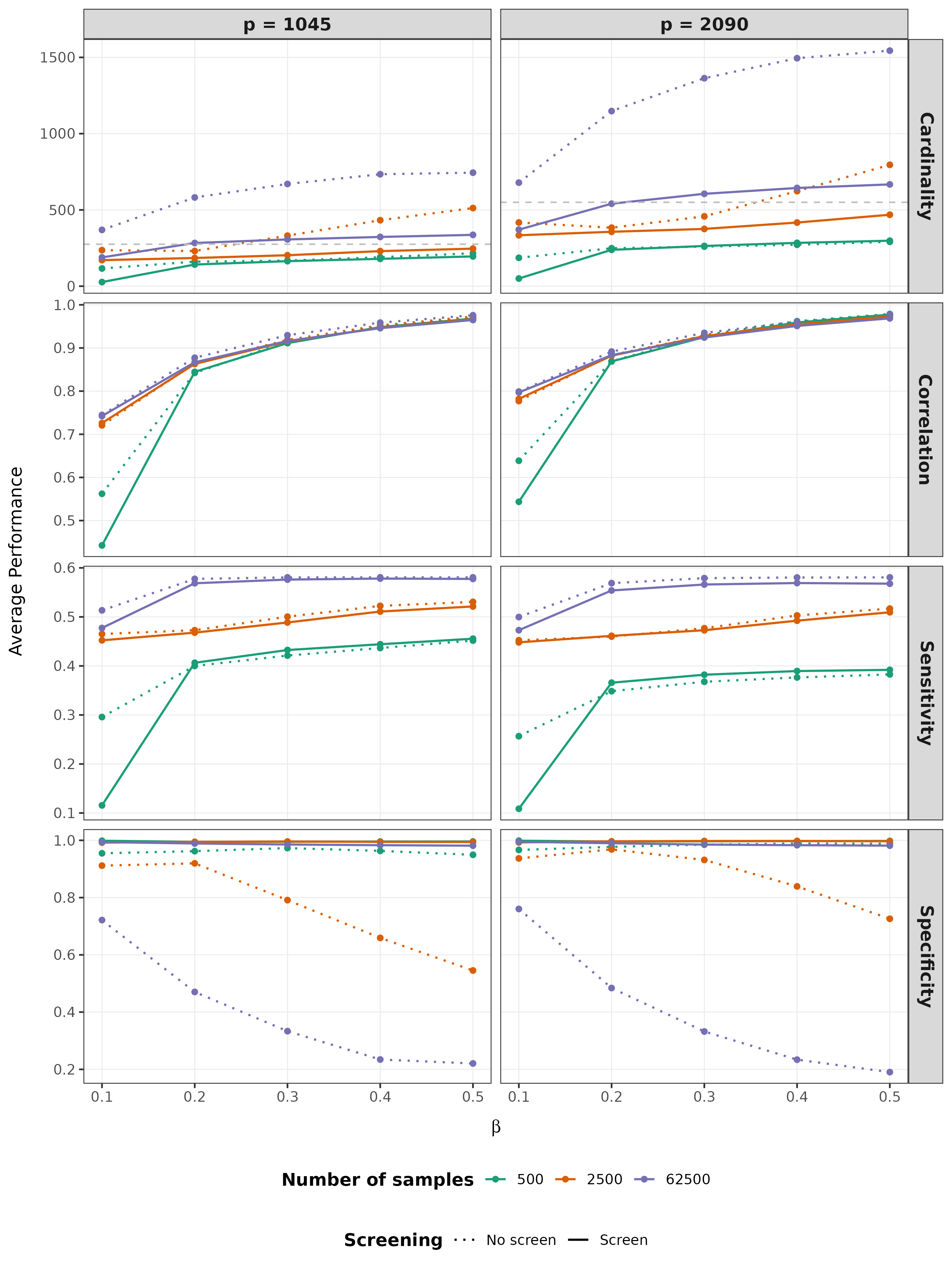}
\caption{Simulation results comparing feature selection strategies across varying sample sizes ($n \in {500, 2500, 62500}$) and feature dimensions ($p \in {1045, 2090}$). Each simulation was based on a LASSO regression model including one exposure, one confounder, and a block-structured set of features comprising 55 ($p = 1045$) or 110 ($p = 2090$) groups of 19 features, with a causal structure as illustrated in Figure \ref{fig:Dagtoy}. Performance was evaluated using six metrics: (row 1) model cardinality (number of selected features), (row 2) correlation between the exposure and the derived feature signature, (row 3) sensitivity to asymptotically selected features, and (row 4) specificity.}\label{fig:supp_scen3_allcrit}
\end{figure}

\begin{figure}[ht]
\centering
\includegraphics[width=0.9\textwidth]{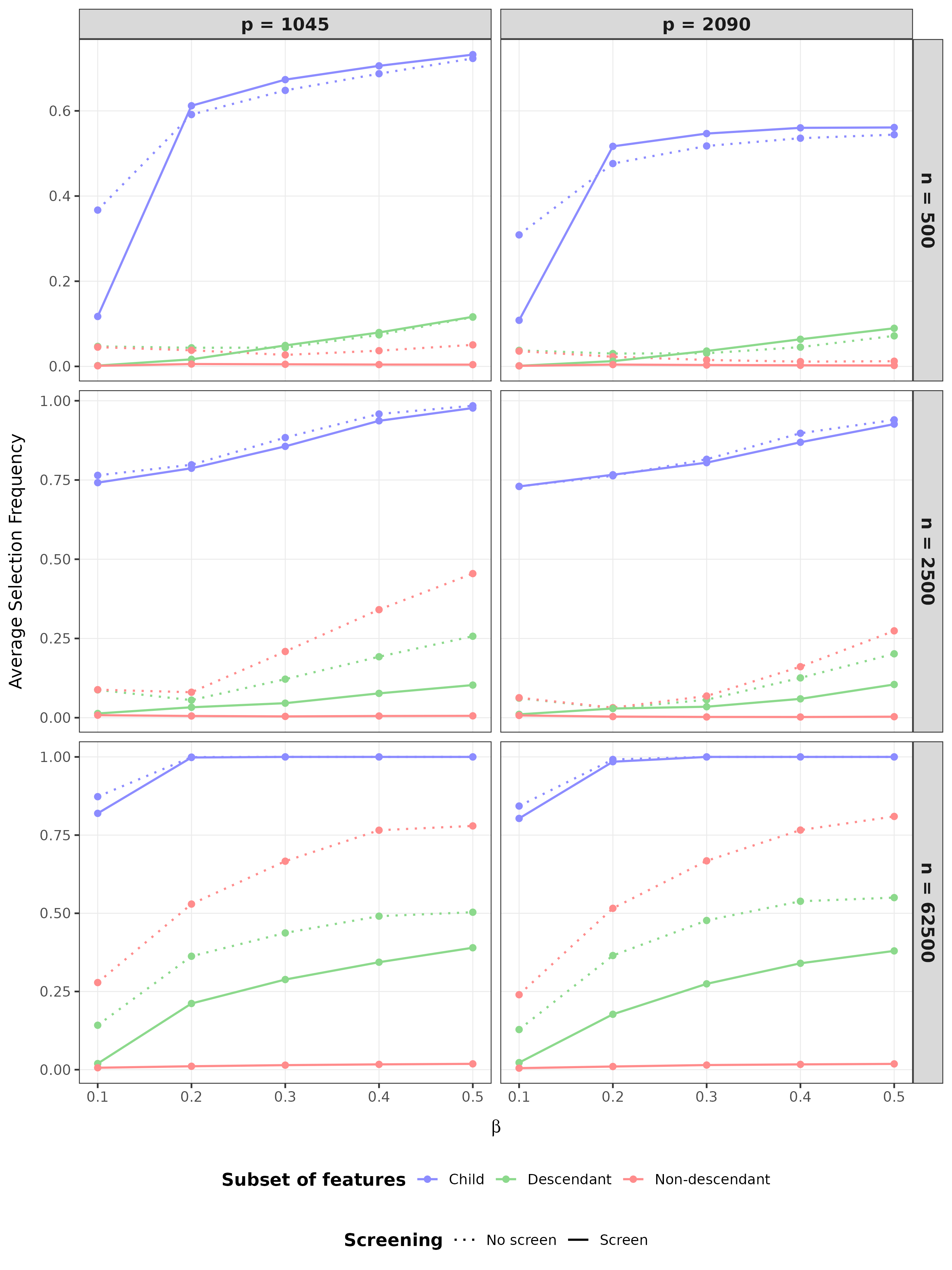}
\caption{Simulation results comparing selection frequency of children, non-child descendants, and non-descendants across varying sample sizes ($n \in {500, 2500, 62500}$) and feature dimensions ($p \in {1045, 2090}$). Each simulation was based on a LASSO regression model including one exposure, one confounder, and a block-structured set of features comprising 55 ($p = 1045$) or 110 ($p = 2090$) groups of 19 features, with a causal structure as illustrated in Figure \ref{fig:Dagtoy}.}
\label{fig:supp_scen3_selfreq}
\end{figure}

\end{appendices}

\end{document}